%% file: main.tex
\documentclass{article}
\usepackage{ijcai20}
\usepackage{color}
\usepackage{times}
\usepackage{soul}
\usepackage{url}
\usepackage[hidelinks]{hyperref}
\usepackage[utf8]{inputenc}
\usepackage[small]{caption}
\usepackage{graphicx}
\usepackage{amsmath}
\usepackage{amsthm}
\usepackage{comment}
\usepackage{booktabs}
\usepackage{indentfirst}
\usepackage{stfloats}
\usepackage[noend]{algpseudocode}
\usepackage{algorithm}
\usepackage{setspace}
\newtheorem{problem}{Problem}
\newtheorem{lemma}{Lemma}
\newtheorem{definition}{Definition}

\title{An Augmented Index-based Efficient Community Search \\for Large Directed Graphs}

\author{
Yankai Chen$^1$\and
Jie Zhang$^2$\and
Yixiang Fang$^3$\footnote{Corresponding Author}\and
Xin Cao$^3$\And
Irwin King$^1$\\
\affiliations
$^1$Department of Computer Science and Engineering, The Chinese University of Hong Kong\\
$^2$School of Computer Science and Engineering, Nanyang technological university\\
$^3$The University of New South Wales\\
\emails
\{ykchen, king\}@cse.cuhk.edu.hk,
zhangj@ntu.edu.sg,
\{yixiang.fang, xin.cao\}@unsw.edu.au
}

\begin{document}

\maketitle

\input{abstract}
\input{intro}
\input{related}
\input{problem}
\input{index}

\input{discussion}
\input{exp}
\input{conclusion}

\section*{Acknowledgments}

\appendix
\input{appdendix}

\newpage
\bibliographystyle{apalike}
\bibliography{ref}

\end{document}

%% file: abstract.tex
\begin{abstract}
Given a graph $G$ and a query vertex $q$, the topic of \emph{community search} (CS), aiming to retrieve a dense subgraph of $G$ containing $q$, has gained much attention. Most existing works focus on undirected graphs which overlooks the rich information carried by the edge directions. Recently, the problem of community search over directed graphs (or CSD problem) has been studied \cite{fang2019effective}; it finds a connected subgraph containing $q$, where the in-degree and out-degree of each vertex within the subgraph are at least $k$ and $l$, respectively. However, existing solutions are inefficient, especially on large graphs. To tackle this issue, in this paper we propose a novel index called \emph{D-Forest}, which allows a CSD query to be completed within the optimal time cost. We further propose efficient index construction methods. Extensive experiments on six real large graphs show that our index-based query algorithm is up to two orders of magnitude faster than existing solutions.
\end{abstract}

%% file: intro.tex
\section{Introduction}
\label{sec:intro}
With the rapid development of information technologies, large graphs are ubiquitous in various areas (e.g., social networks and biological science)~\cite{ying2018graph}. Finding communities over these graphs is fundamental to many real applications, such as event organization, recommendation, and network analysis. In recent years, the topic of \emph{community search} (CS) has gained much attention (e.g., \cite{sozio2010community,cui2014local,huang2014querying,fang2019survey}),
which aims to find dense communities containing the query vertex $q$ from a graph $G$ in an online manner.

Earlier CS works (e.g., ~\cite{sozio2010community,cui2013online,cui2014local,barbieri2015efficient}) mainly focus on undirected graphs, where the graph edges do not have directions.
They often require the community to be a connected subgraph satisfying a particular metric of structure cohesiveness (e.g., each vertex within the community has a degree of $k$ or more).
Later on, the attributes (e.g., keywords \cite{fang2017effective}, locations \cite{fang2017spatial}, profiles \cite{PCS19}, and influence values \cite{li2015influential}) of graphs have also been considered for CS.
However, all these works ignore the directions of edges in the graphs. As pointed out in~\cite{malliaros2013clustering,zhang2014proposed}, the ignorance of directions of edges, failing to capture asymmetric relationships implied by the directions, may lead to noise and inaccurate results.

\begin{figure}[]
\centering
     \includegraphics[width=1\columnwidth]{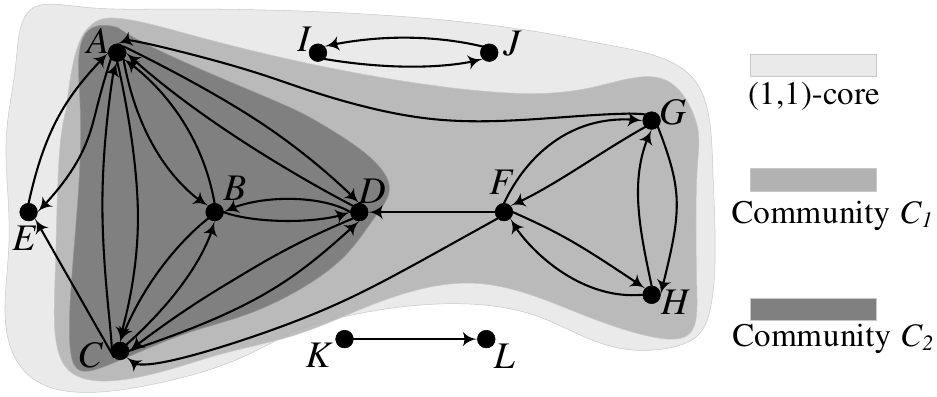}
    \caption{A directed graph.}
    \label{fig:directedG}
\end{figure}

To remedy the issue above, Fang et al.~\shortcite{fang2019effective} studied the problem of CS over directed graphs, or CSD problem. Specifically, given a query vertex $q$ of a directed graph $G$, and two integers $k$ and $l$, it aims to find a maximum connected subgraph containing $q$, where the in-degree and out-degree of each vertex within the subgraph are at least $k$ and $l$, respectively.
For example, in Figure \ref{fig:directedG}, let $q$=$B$ and $k$=$l$=2; then, the subgraph $C_1$ will be returned; if $k$=$l$=3, then the subgraph $C_2$ is the answer.
The minimum degree constraints have also been used in the ($k$, $l$)-core \cite{giatsidis2013d}, or the maximum subgraph in which each vertex's in-degree and out-degree are at least $k$ and $l$, respectively. Note that the ($k$, $l$)-core may not be a connected subgraph.

To answer a CSD query, an online method is to iteratively peel vertices that do not satisfy the degree constraints until finding the subgraph satisfying both the connectivity and minimum degree constraints. Obviously, this iterative approach could be costly. To improve efficiency, Fang et al.~\shortcite{fang2019effective} develop indexes, which pre-compute all the ($k$, $l$)-cores and then organize them compactly. Given a CSD query, they first find the ($k$, $l$)-core using the indexes, and then compute the maximum connected subgraph containing $q$ from the ($k$, $l$)-core. Clearly, since this maximum connected subgraph is often much smaller than the ($k$, $l$)-core, their solutions are inefficient if the ($k$, $l$)-core is very large. For example, on a graph with about 4 billion edges, they may take over 33 minutes to answer 200 queries as shown by our experiments.

To facilitate efficient CSD queries, in this paper we develop a novel index structure, called D-Forest, which not only compactly organizes all the ($k$, $l$)-cores, but also takes the connectivity into consideration. Specifically, we first compute all the ($k$, $l$)-cores sequentially where $k$ and $l$ range from 0 to their maximum values. Then, for each value of $k$, we compute their connected components and organize them into a tree structure, where each tree node corresponds to a connected ($k$, $l$)-core. As a result, the index is a forest, consisting of a list of trees.
Given the D-Forest, to answer a CSD query, we first find the $k$-th tree and then return the connected ($k$, $l$)-core containing $q$. Clearly, the query takes the optimal query time cost, i.e., $O(|C|)$, where $C$ is the set of vertices in the community. 
We also discuss index maintenance for dynamic graphs and how to use D-Forest to answer the query of a variant of the CSD problem, i.e., SCSD \cite{fang2019effective}. 
Experiments on six real large graphs show that our index construction process takes comparable time cost with those of existing algorithms, and our query algorithm takes only 14.4 seconds to answer 200 queries on a graph with around 4 billion edges.

In summary, our main contributions are as follows:
\begin{itemize}
\item We develop a novel index D-Forest, based on which a CSD query can be completed in optimal time cost.
\vspace{-0.05in}
\item To build the D-Forest, we propose a basic algorithm, and an advanced algorithm by introducing an auxiliary data structure called Core-based Union-Find (or CUF).
\vspace{-0.05in}
\item We discuss the index maintenance for dynamic graphs, and how to answer the SCSD query using D-Forest.
\vspace{-0.05in}
\item We perform extensive experiments on six  real large graphs; results show that our query algorithm is up to two orders of magnitude faster than existing solutions.
\end{itemize}

%% file: related.tex
\section{Related Works}
\label{sec:related}
\textbf{Community detection (CD).} Generally, CD aims to detect all the communities from an entire graph. Earlier classic solutions (e.g., \cite{newman2004finding,newman2004fast,fortunato2010community}) rely on edge-based analysis (e.g., modularity maximization) to discover these communities. However, most of them focus on undirected graphs.
Recent works start to detect communities from directed graphs.
In \cite{leicht2008community,kim2010finding}, the concept of modularity maximization~\cite{newman2004fast} is extended for CD on directed graphs.
In \cite{lancichinetti2009benchmarks}, authors introduced new benchmark graphs to test CD methods over directed graphs.
Yang et al. \shortcite{yang2010directed} introduced a new stochastic block model called PPL to find communities in directed graphs; they also detected overlapped communities in directed graphs \cite{yang2014detecting}.
Besides, there are also some local CD methods (e.g., \cite{flake2000efficient,ning2016local}).
A recent survey of CD solutions on directed graphs can be found in~\cite{malliaros2013clustering}.
However, as pointed out in~\cite{fortunato2016community}, CD is an ill-defined problem with no clear universal definition of the objects that one should be looking for.
However, these CD methods are often time consuming, especially on large graphs, and also it is not clear how they can be adapted for online CS.

\textbf{Community search (CS).} CS finds the community from a large graph in a fast and online manner, based on a query request.
To measure the structure cohesiveness of communities, the $k$-core metric is often employed, requiring that each vertex of the community should have a degree of $k$ or more, where $k$ is a given integer ~\cite{batagelj2003m,dorogovtsev2006k,sozio2010community,cui2014local,li2015influential,fang2017effective,fang2017spatial,PCS19}.
Other cohesiveness metrics have also been considered for CS, such as $k$-clique~\cite{yang2011social,cui2013online},
$k$-truss~\cite{huang2014querying,huang2015approximate,huang2016truss,akbas2017truss,huang2017attribute,ebadian2019fast} and $k$-ECC~\cite{hu2017minimal}, pagerank-based~\cite{andersen2006communities,kloumann2014community}, etc.
A survey of CS over graphs can be found in~\cite{fang2019survey}.
However, most of these works focus on undirected graphs. A recent work \cite{fang2019effective} has studied CS over directed graphs, but its solutions are still inefficient for large graphs, calling for more efficient CS approaches.

%% file: problem.tex
\section{Problem Definition}
\label{sec:prob}
We consider a directed graph $G(V,E)$ with a vertex set $V$ and an edge set $E$. The sizes of $V$ and $E$ are respectively denoted by $n$ and $m$. The in-degree and out-degree of a vertex $v$ in $G$ are denoted by $deg_G^{in}(v)$ and $deg_G^{out}(v)$. Next, we introduce the core model on directed graphs.

\begin{definition}[$(k,l)$-core~\cite{giatsidis2013d}]
Given a directed graph $G(V,E)$ and two non-negative integers $k$ and $l$, the ($k$,$l$)-core of $G$ is the largest subgraph $G'$ of $G$, such that $\forall v \in G'$, $deg^{in}_{ G'}(v) \geq k$ and $deg^{out}_{ G'}(v) \geq l$.
\end{definition}

In the $(k,l)$-core, each vertex has at least $k$ in-neighbours and $l$ out-neighbours, so it is well engaged in the subgraph especially when $k$ and $l$ are large. This implies that the $(k,l)$-core is a cohesive subgraph, and thus can be used to model the cohesiveness of the community~\cite{fang2019effective}. However, the ($k$, $l$)-core may not be a connected subgraph, so the connectivity constraint should be further imposed to model the community. Note that for simplicity, we denote a connected ($k$, $l$)-core by ($k$, $l$)-$\widehat{core}$.

Based on the discussions above, Fang et al.~\shortcite{fang2019effective} formally introduced the problem of \underline{C}ommunity \underline{S}earch over \underline{D}irected graphs (CSD) problem:

\begin{problem}[CSD problem~\cite{fang2019effective}]
Given a directed graph $G(V,E)$, a query vertex $q$, and two positive integers $k$ and $l$, return the ($k$, $l$)-$\widehat{core}$ containing $q$.
\end{problem}

For example, in Figure \ref{fig:directedG}, the (1, 1)-core, (2, 2)-core, and (3, 3)-core are marked in three different colors, where the (1, 1)-core has three connected components. If $q$=$B$ and $k$=$l$=3, then the (3, 3)-$\widehat{core}$ $C_2$ is returned as the community.

%% file: index.tex
\section{Our Index-based Approach}
\label{sec:basic}
To enable efficient CSD queries, in this paper we propose a novel index, called D-Forest, which allows the targeted community to be retrieved directly without examining the ($k$, $l$)-core. As a result, the query time cost is optimal. Meanwhile, the index is space efficient since it takes $O(m)$ space cost. In the following sections, we first give an overview of D-Forest, and then present two algorithms to build D-Forest.

\subsection{Index Overview}
\label{sec:overview}
We begin with an interesting lemma.

\begin{lemma}[\cite{fang2019effective}]
\label{lm:nest}
Given a directed graph $G$, for any ($k$, $l$)-core with $l\textgreater0$, it is a subgraph of the ($k$, $l$$-$$1$)-core, i.e., the ($k$, $l$)-core is nested within the ($k$, $l$$-$$1$)-core.
\end{lemma}

By Lemma \ref{lm:nest}, we can conclude that for any specific value of $k$, all the ($k$, $l$)-cores where $l$ ranges from 0 to its maximum value can be organized into a chain such that each one is nested within its previous one.
Similarly, the nested property above holds for the connected components of ($k$, $l$)-cores; that is, for any ($k$, $l$)-$\widehat{core}$ with $l\textgreater0$, it is nested within a particular ($k$, $l$--1)-$\widehat{core}$, so we can get a chain for each ($k$, $l$)-$\widehat{core}$. Consequently, for each value of $k$, we can build a tree structure called $k$-tree, by hierarchically organizing all these chains, such that each subtree corresponds to a ($k$, $l$)-$\widehat{core}$.

For example, all ($k$, $l$)-$\widehat{core}$s in Figure \ref{fig:directedG} can be organized into 4 trees, as depicted in Figure \ref{fig:index}, where the $k$-tree is built for the value of $k$ and each subtree contains all the vertices of a particular ($k$, $l$)-$\widehat{core}$. For instance, in the $1$-tree, the subtree rooted at node~\footnote{In this paper, we use ``node'' to mean the ``index tree node''.} $p$, as shown in the dashed box, contains vertices $\{F,G,H\}$ and $\{A,B,C,D\}$, which are the vertices in the $(1,2)$-$\widehat{core}$. Note that the number attached in each node indicates the value of $l$ for the corresponding ($k$, $l$)-$\widehat{core}$ and we use the root node $t$ to keep the tree shape.
 
\begin{figure}[htp]
    \centering
     \hspace{-0.2cm}\includegraphics[width=1\columnwidth]{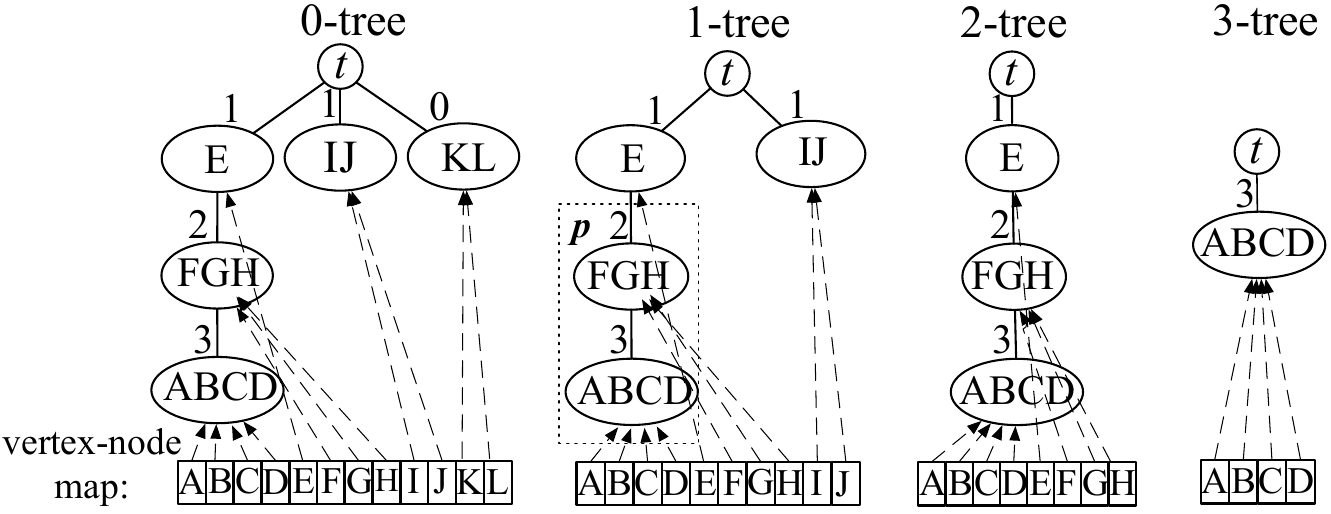}
    \caption{An example D-Forest index.}
    \label{fig:index}
\end{figure}

To summarize, in the $k$-tree, each node has four elements:

\begin{itemize}
\item $parent$: a pointer to its parent node;
\vspace{-0.05in}
\item $childList$: a list of pointers to its child nodes;
\vspace{-0.05in}
\item $vSet$: a set of vertices, which are in the ($k$, $l$)-$\widehat{core}$ but not in the ($k$, $l$--1)-$\widehat{core}$;
\vspace{-0.05in}
\item $coreNum$: the value of $l$ of the ($k$, $l$)-$\widehat{core}$, which corresponds to the subtree rooted at this node.
\end{itemize}

To enable efficient locating of the ($k$, $l$)-$\widehat{core}$, for each tree, we build an auxiliary map, where each key is a vertex and its value points to the node containing the vertex in the tree. All proofs of lemmas are attached in Appendix~\ref{appendixA}. 

\begin{lemma}[Space cost]
\label{lm:space}
Given a directed graph $G$, its D-Forest takes $O(m)$ space.
\end{lemma}

\noindent{\bf\underline{$\bullet$ Query algorithm.}} To answer a CSD query, we can first use the auxiliary map to locate the tree node $p$ which contains $q$, then find the root of the subtree corresponding to the ($k$, $l$)-$\widehat{core}$ containing $q$, and finally return the set $C$ of vertices in the subtree. We denote the query algorithm above by {\tt IDX-Q}.

\begin{lemma}[Query cost]
\label{lm:query}
Given a D-Forest, {\tt IDX-Q} completes in the optimal time and space cost, i.e., $O(|C|)$.
\end{lemma}


Next, we present two index construction algorithms, which work in the top-down and bottom-up manners, respectively.

\subsection{A Top-Down Index Construction Method}
\label{sec:basic}
In \cite{fang2019effective}, an efficient algorithm of decomposing ($k$, $l$)-cores for a graph is developed. It enumerates $k$ from 0 to its maximum value $k_{\max}$, and for each $k$, it computes all the ($k$, $l$)-cores where $l$ ranges from 0 to its maximum value $l_{\max}$. Since D-Forest is comprised of ($k_{\max}$$+$$1$) trees, we can also enumerate $k$ from 0 to $k_{\max}$, and for each specific $k$, we build the $k$-tree by the following the steps:

\begin{enumerate}
\item compute the ($k$, $0$)-core, create a node with a set of vertices in ($k$, $0$)-core, and initialize $l$=$1$;
\vspace{-0.05in}
\item compute all the ($k$, $l$)-$\widehat{core}$s from the ($k$, $l$$-$$1$)-$\widehat{core}$s;
\vspace{-0.05in}
\item for each ($k$, $l$)-$\widehat{core}$, create a node $p$ ($p.vSet$ contains vertices in ($k$, $l$)-$\widehat{core}$, $p$.$coreNum$=$l$), link $p$ to its parent node $p'$, and update $p'.vSet$ as $p'.vSet\backslash p.vSet$;
\vspace{-0.05in}
\item increase $l$ by 1, and repeat steps 2 and 3 until $l$ reaches its maximum value.
\end{enumerate}

Since the above method builds trees of D-Forest in a top-down manner, we denote it by {\tt TopDown}, the pseudocodes of which are shown in Algorithm~\ref{alg:basicindex}.

\begin{algorithm}[]
\caption{Index construction algorithm: {\tt TopDown}.}
\label{alg:basicindex}
\footnotesize{
\algrenewcommand{\algorithmiccomment}[1]{\hskip3em$//$ #1}
\begin{algorithmic}[1]
    \Function{build($G$)}{}
        \State $\mathcal F\gets\emptyset$;
        \For{$k \gets  0$ to $k_{\max}$}
            \State $map \gets \emptyset$, compute $(k,0)$-core and create root node $t$;
             \For{$l \gets  1$ to $l_{\max}$}
             	\State $S \gets$ compute $(k,l)$-$\widehat{core}$s from $(k,l-1)$-$\widehat{core}$s;
             	\For{each $(k,l)$-$\widehat{core}$ in $S$}
             		\State create a node $p$ and $p.coreNum \gets l$;
             		\State $p.vSet \gets$ vertices in this $(k,l)$-$\widehat{core}$;
             		\State locate $p$'s parent node $p'$, link $p'$ and $p$;
             		\State $p'vSet \gets p'vSet\backslash p.vSet$;
             		\For{$v$ in $p'vSet$} $map.$put($v, p'$); \EndFor
             	\EndFor	
            \EndFor
            \State$\mathcal F$.add($\langle t$, $map\rangle$); 
        \EndFor
        \State\Return $\mathcal F$;
    \EndFunction
\end{algorithmic}}
\end{algorithm}

\begin{lemma}
\label{lemma:topdown}
Given a directed graph $G$, the time cost of building D-Forest using {\tt TopDown} is $O(m^2)$.
\end{lemma}

\input{advanced}

%% file: advanced.tex
\subsection{A Bottom-Up Index Construction Method}
\label{sec:adv}
While {\tt TopDown} is easy to implement, it may suffer from the low efficiency issue, as shown in Lemma \ref{lemma:topdown}. To further improve the efficiency, we propose another more efficient index construction method {\tt BottomUp}, by introducing an auxiliary data structure called \underline{C}ore-based \underline{U}nion-\underline{F}ind (or CUF), which builds the trees in a bottom-up manner. Next, we first give an overview of {\tt BottomUp}, and then introduce the details.

\noindent{\bf\underline{$\bullet$ Overview of {\tt BottomUp}.}} Unlike {\tt TopDown}, {\tt BottomUp} enumerates the values of $k$ from $k_{\max}$ to 0, and builds each tree in a bottom-up manner (i.e., create leaf nodes at first and root node at last). Meanwhile, when building the $k$-tree, it exploits the information generated in building ($k$+$1$)-tree. 

\begin{algorithm}[]
\caption{Index construction algorithm: {\tt BottomUp}.}
\label{alg:index}
\footnotesize{
\algrenewcommand{\algorithmiccomment}[1]{\hskip3em$//$ #1}
\begin{algorithmic}[1]
    \Function{build($G$)}{}
        \State $\mathcal F\gets\emptyset$, $pre[\text{ }] \gets \emptyset$, $cur[\text{ }] \gets \emptyset$, $\Psi \gets \emptyset$;
        \For{$k \gets  k_{\max}$ to $0$}
            \State $map \gets \emptyset$, $\mathcal{P} \gets \emptyset$;
            \State $cur[\text{ }] \gets \Call{decompose($G$, $k$)}{}$;
            \State $V_0, \cdots, V_{l_{\max}}$ $\gets$ group vertices of $cur[\text{ }]$; 
            \For{$l \gets  l_{\max}$ to $0$}
                \State \textsc{BuildALevel}($k$, $l$, $V_l$, $pre[\text{ }]$, $cur[\text{ }]$, $map$, $\mathcal{P}$, $\Psi$);
            \EndFor
            \State  create root node $t$ and link to nodes in $\mathcal P$;
            \State $pre[\text{ }] \gets cur[\text{ }]$, $\mathcal F$.add($\langle t$, $map\rangle$);
        \EndFor
        \State\Return $\mathcal F$;
    \EndFunction
\end{algorithmic}}
\end{algorithm}

Algorithm~\ref{alg:index} outlines {\tt BottomUp}.
We first initialize an empty forest $\mathcal F$, two arrays $pre[\text{ }]$, $cur[\text{ }]$, where $pre[v]$ and $cur[v]$ are supposed to keep the maximum value of $l$ such that there is a ($k$, $l$)-core containing $v$. We also initialize the CUF data structure $\Psi$ which will be introduced later (line 2). 
 Then, we enumerate $k$ from $k_{\max}$ to 0 and for each $k$, we compute all the ($k$, $l$)-cores (lines 3-5), by using the algorithm in \cite{fang2019effective}. We initialize the vertex-node map $map$, and a set $\mathcal P$ for keeping the generated nodes for each level of $k$-tree (line 4). 
After that, we group vertices into a list of sets, such that $V_l$ contains vertices which are in the ($k$, $l$)-core but not in the ($k$, $l$--1)-core (line 6).
Next, we build the $k$-tree by invoking \textsc{BuildALevel} to create nodes in the $l$-th level of $k$-tree where $l$ ranges from $l_{\max}$ to $0$ (lines 7-8).
Finally, we create the root node with $\mathcal {P}$, and update $pre[\text{ }]$ and $\mathcal F$ (lines 9-10).

\noindent{\bf\underline{$\bullet$ Overview of function \textsc{BuildALevel}.}}
Given nodes in the ($l$+1)-th level of $k$-tree, the function \textsc{BuildALevel} creates nodes in the $l$-th level and links them to the nodes in the ($l$+1)-th level. Since each node corresponds to a ($k$, $l$)-$\widehat{core}$, a naive method to check the connectivity and create the node will take $O(m)$ time to re-explore the graph, i.e., executing steps 2 and 3 of {\tt TopDown}. Consequently, using this naive method totally takes $O(l_{\max}\cdot m)$ to build the $k$-tree, which is the same as that of {\tt TopDown}.

To improve the efficiency, we propose a novel data structure, called Core-based Union-Find (or CUF), which allows the three key steps of \textsc{BuildALevel} to be done efficiently:
(1) verifying the connectivity,
(2) memorizing the connectivity, and 
(3) linking nodes.
In the following sections, we first introduce the CUF data structure, and then present our CUF-based \textsc{BuildALevel}, which allows the $k$-tree to be built in $O(\alpha(n)\cdot m)$ time, where $\alpha(n)$ is the inverse Ackermann function and $\alpha(n) < 5$ for any practical value of $n$.

\noindent{\bf\underline{$\bullet$ CUF data structure.}}
CUF is extended from classic Union-Find (UF) Forest\footnote{\url{https://en.wikipedia.org/wiki/Disjoint-set_data_structure}}, which can efficiently verify the graph connectivity and partition vertices into different connected components. In the classic UF, each vertex has 2 elements, i.e., $rank$ and $parent$, and the UF has 3 functions, i.e., \textsc{makeSet}, \textsc{find} and \textsc{union}, where \textsc{makeSet} makes preparation for each vertex, \textsc{find} returns the representative member of the component to which the vertex belongs, and \textsc{union} merges two disjoint components as one. By using the classic UF, given a ($k$, $0$)-core, we can verify the connectivity and sequentially find all ($k$, $l$)-$\widehat{core}$'s by varying $l$ from $l_{\max}$ to $0$, and then build all the levels of the $k$-tree accordingly. However, classic UF may have two main limitations.

\begin{algorithm}[]
\caption{Functions of the CUF data structure.}
\label{alg:CUF}
\footnotesize{
\algrenewcommand{\algorithmiccomment}[1]{\hskip3em$//$ #1}
\begin{algorithmic}[1]
    \Function{makeSet($v$)}{}
        \State $v.rank\gets 0$, \ $v.parent \gets v$;
        \State \underline{$v.hook\gets v$, $v.group \gets v$;}
    \EndFunction

    \Function{find($v$)}{}
        \If {$v.parent$ = $v$}
            \State $v.parent\gets$ \Call{find($v.parent$)}{};
        \EndIf
        \State \Return $v.parent$;
    \EndFunction

    \Function{union($u$, $v$, $cur[\text{ }]$)}{}
        \State $r_u \gets$ \Call{find($u$)}{}, $r_v\gets$  \Call{find($v$)}{};
        \If {$r_u \neq r_v$}

        \If{$r_u.rank < r_v.rank$} \Call{swap($r_u$, $r_v$)}{}; \EndIf
        \State $r_v.parent\gets r_u$;
        \If{$r_u.rank = r_v.rank$}   
            \State $r_u.rank\gets r_u.rank + 1$; \EndIf
        \If{\underline{$cur[r_u.group] < cur[r_v.group]$ }}
            \State \underline{$r_u.group \gets r_v.group$;}
        \EndIf
        \EndIf
    \EndFunction

    \Function{updateCuf($V$, $cur[\text{ }]$)}{}        
        \For{\underline{$u \in V$}}
            \State \underline{$r \gets$ \Call{find($v$)}{};}
            \State \underline{$v.group \gets r.group$;}
            \If{\underline{$cur[r.hook] > cur[v]$}} \underline{$r.hook \gets v$;} \EndIf
            
        \EndFor
    \EndFunction

\end{algorithmic}}
\end{algorithm}

One limitation is that for a new node $p$ in the $l$-th level, classic UF can not efficiently find $p$'s all child nodes and link them up. 
As observed in Section~\ref{sec:overview}, each subtree below the $l$-th level corresponds to a particular ($k$, $l'$)-$\widehat{core}$ where $l'\textgreater l$. This means that to find $p$'s child nodes, we can first find all ($k$, $l'$)-$\widehat{core}$'s that are connected by vertices in $p.vSet$ and then link the root nodes of corresponding subtrees to $p$. 
To efficiently locate these subtrees, we assign another element $hook$ to directly indicate these root nodes. For example, for a vertex $v$ in $p.vSet$, if $v$'s neighbour $u$ is contained in a ($k$, $l'$)-$\widehat{core}$, we can locate the root node of the corresponding subtree by referring $hook$ and then link it to $p$.

The other limitation is that once the $l$-th level of the ($k$+$1$)-tree is constructed, the connectivity of the corresponding ($k$+$1$, $l$)-core is verified. When building the $l$-th level of the $k$-tree, we have to traverse the corresponding ($k$, $l$)-core and verify its connectivity from scratch. However, from Lemma~\ref{lm:nest}, ($k$+$1$, $l$)-core is a subgraph of ($k$, $l$)-core, which implies that the connectivity of this ($k$+$1$, $l$)-core will be verified again. 
To cut off this redundant computation, we assign an additional element $group$ in CUF structure to ``memorize'' the particular ($k$, $l$)-$\widehat{core}$ to which each vertex used to belong. For instance, if vertex $v$ is included in a certain ($k$+$1$, $l$)-$\widehat{core}$, $v.group$ will be marked; and when processing the ($k$, $l$)-core, by checking $v.group$, we would quickly know that $v$ should be gathered together with others who share the same value.

Algorithm~\ref{alg:CUF} summarizes all CUF functions and underlines our contribution.
To maintain the two additional elements, we propose a new function \textsc{updateCuf}. As shown in Algorithm~\ref{alg:CUF}, for each vertex $v$, we find the representative member $r$ of the particular ($k$, $l$)-$\widehat{core}$ including $v$, i.e., \textsc{find($v$)} (lines 18-19). Then we update $v.group$ as $r.group$ and set $r.hook$ as $v$ if $v$ has smaller value in $cur[\text{ }]$ (lines 20-21).
 
\begin{lemma}[Space cost]
Given a directed graph $G$, the CUF data structure of all vertices costs $O(n)$ space.
\end{lemma}

\begin{lemma}[Time cost of CUF functions]
\label{lm:cuf}
\textsc{makeSet} takes $O(1)$ time; for \textsc{union} and \textsc{find}, the amortized time per operation is $O(\alpha(n))$; \textsc{updateCuf} takes $O(\alpha(n) \cdot |V|)$ time.
\end{lemma}

\noindent{\bf\underline{$\bullet$ Details of function \textsc{BuildALevel}.}}
Algorithm~\ref{alg:build} shows the details. Firstly, we find root nodes of subtrees to be linked to new nodes in this level. For each vertex $v$, we initialize a set $S_v$; we visit $v$'s neighbour $u$ to find the root $p'$ of the subtree including $u$ and collect it in $S_v$ (lines 2-8). 
Then we use CUF to verify the subgraph connectivity for this level. We initialize a set $V'$ to collect vertices if they has been previously processed in the $l$-th level of the ($k$+$1$)-tree (lines 9-13). For vertices in $V'$, we directly use their $group$ to achieve a quick \textsc{union}; for others, we visit their neighbours to check the connectivity by invoking \textsc{batchUnion} (lines 15-16, 24-28). 
After that, for each vertex set $C_i$ sharing the same CUF root, we create the node $p$ and update $\mathcal P$ (lines 17-19). For each vertex $v$ in $C_i$, we update the vertex-node map $map$ and link child nodes in $S_v$ to $p$ (lines 20-22).
Finally, we update CUF for the constrcution of next level (line 23).  

\begin{algorithm}[htp]
\caption{Process vertices to create tree nodes.}
\label{alg:build}
\footnotesize{
\algrenewcommand{\algorithmiccomment}[1]{\hskip3em$//$ #1}
\begin{algorithmic}[1]
    \Function{BuildALevel($k$, $l$, $V_l$, $pre[\text{ }]$, $cur[\text{ }]$, $map$, $\mathcal P$, $\Psi$)}{}
        \State initialize $S_{v_1}$, $\cdots$, $S_{v_i}$, $\cdots$ for vertices $v_1$, $\cdots$, $v_i$, $\cdots \in V_l$; 
        \For{$ v \in V_l$}
            \For{$u \in N(v)$}
                    \If{$cur[u] > cur[v]$}
                        \State $r_u \gets$ $\Psi$.\Call{find($u$)}{};
                        \State $p' \gets map.$get($r_u.hook$);
                        \State $S_v$.add($p'$), \ $\mathcal P$.delete($p'$);
                    \EndIf
                \EndFor 
        \EndFor

        \State $V' \gets \emptyset$;
         \For{$v \in V_l$}
            \If{$k \neq k_{\max}$ and $pre[v] = l$}
                \State $v.parent \gets v$, $v.hook \gets v$;
                \State  $v.rank \gets 0$, $V'.$add($v$);
            \Else
                \ \ $\Psi$.\Call{makeSet($v$)}{};
            \EndIf

         \EndFor
        \State \Call{batchUnion($V_l\backslash V'$, $cur[\text{ }]$, $\Psi$)}{};

        \For{$v \in V'$}
            $\Psi$.\Call {union($v$, $v.group$, $cur[\text{ }]$)}{};
        \EndFor

        \For{each set $C_i \subseteq V_l$ with the same CUF root}
            \State $p \gets$ create tree node by using $l$ and $C_i$;
            \State $\mathcal P$.add($p$);
            \For{each $v \in C_i$}
                \State $map$.put($v$, $p$);
                \State link nodes in $S_v$ with $p$; 
            \EndFor
        \State $\Psi$.\Call{updateCuf($C_i$, $cur[\text{ }]$)}{};
        \EndFor
    \EndFunction
     \Function{batchUnion($V$, $cur[\text{ }]$, $\Psi$)}{}
        \For {$v \in V$}
                \For{$u \in N(v)$}
                    \If{$cur[u] \geq cur[v]$} 
                    \State $\Psi$.\Call{union($u$, $v$, $cur[\text{ }]$)}{};\EndIf
                \EndFor
        \EndFor
    \EndFunction 
\end{algorithmic}}
\end{algorithm}

\begin{lemma}
\label{lm:bottomup}
Given a directed graph $G$, the time cost of building D-Forest using {\tt BottomUp} is $O(\alpha(n)\cdot m \cdot \sqrt{m})$.
\end{lemma}

%% file: discussion.tex
\section{Discussions}
\label{sec:discussion}
\vspace{-0.02in}
In this section, we discuss how to solve a variant of the CSD problem using D-Forest. We also discuss index maintenance techniques.

\vspace{-0.05in}
\subsection{SCSD Problem}
\label{sec:scsd}
\vspace{-0.02in}
The community returned by a CSD query may not be a strongly connected component (SCC), i.e., each pair of vertices in are mutually reachable. To search the SCC-based communities, Fang et al. introduced the SCSD problem:

\begin{problem}[SCSD problem~\cite{fang2019effective}]
Given a directed graph $G(V,E)$, a query vertex $q$, and two positive integers $k$ and $l$, return a subgraph $G'$ such that $G'$ is an SCC containing $q$ and $\forall v \in G'$, $deg^{in}_{G'}(v) \geq k$ and $deg^{out}_{G'}(v) \geq l$.
\end{problem}

The SCSD problem can be readily solved by using D-Forest. Specifically, we first retrieve the ($k$, $l$)-$\widehat{core}$ $G'$ containing $q$ from the D-Forest index. Then, we compute the SCC $C_q$ containing $q$ from $G'$. After that, if $C_q$ does not satisfy the minimum degree constraints, we compute a ($k$, $l$)-$\widehat{core}$ containing $q$ from $C_q$ and then repeat the above steps to find the answer. Note that to compute SCCs of a graph, we can use Kosaraju's algorithm~\cite{hopcroft1983data}, which only takes linear time. We denote this algorithm by {\tt IDX-SQ}, and omit its pseudocodes for lack of space.

\vspace{-0.05in}
\subsection{Index Maintenance}
\label{sec:maintain}
\vspace{-0.02in}
In practice, graphs are often dynamically updated, so it is necessary to update D-Forest accordingly. A simple method is to rebuild D-Forest from scratch when the graph updates. This, however, is very costly, and thus it is desirable to develop efficient index maintenance techniques. As shown in~\cite{li2013efficient,fang2017effective}, when an edge is inserted or deleted, only a few ($k$, $l$)-cores change while others remain unchanged, so it is possible to update D-Forest efficiently.
Specifically, when a new edge ($u$, $v$) is added in the ($k$, $l$)-core, $u$'s out-degree and $v$'s in-degree will increase by 1. To update D-Forest, we execute three steps:
(1) for the vertex $u$, if the increase of $u$'s out-degree enables $u$ to be included in the ($k$, $l$+1)-core, we locate the node containing $u$ in the $k$-tree and then move $u$ down to the node corresponding to this ($k$, $l$+1)-core;
(2) for the vertex $v$, similarly, if the increase of $v$'s in-degree enables $v$ to be grouped into the ($k$+1, $l$)-core, then we add $v$ to the node corresponding to this ($k$+1, $l$)-core;
(3) after including a vertex into another core, the connectivity of some connected components of the cores may change, so we need to merge the subtrees corresponding to the newly updated subgraph.
For the case of edge deletion, we can update the index by taking the inverse process of above steps. Note that the case of vertex update can be regarded as sequentially inserting or deleting a list of edges.

%% file: exp.tex
\section{Experiments}
\label{sec:exp}

\subsection{Setup}
\label{sec:setup}

We use six real large directed graphs and summarize their statistics in Table~\ref{tab:setup}, where $d$ is the average degree of vertices. The Twitter dataset~\footnote{\url{https://www.isi.edu/lerman/downloads/twitter/twitter2010.html}} is collected by Kristina Lerman. The eu-2015, arabic, it-2004, sk-2005 and uk-2007 datasets~\cite{boldi2018bubing,BoVWFI,chen2021efficient} are available in the website\footnote{\url{http://law.di.unimi.it/datasets.php}}. We implement all the algorithms in Java, and run the experiments on a Linux machine with Ten-core Intel E7-4820 V3 CPU@1.90GHz and 300GB memory.

\begin{table}[hbtp]
    \centering \scriptsize
    \caption{Datasets in our experiments.}
    \label{tab:setup}
    \setlength{\tabcolsep}{2.6mm}{
    \begin{tabular}{c|c|c|c|c|c}
        \hline
         {\bf Graph} &\multicolumn{1}{c|}{\textbf{$n$}}
         &\multicolumn{1}{c|}{\textbf{$m$}}
         &\textbf{$d$}
         &\textbf{$k_{\max}$}
         &\textbf{$l_{\max}$}\\
          \hline
          \hline
         Twitter    &   699,986     &  36,743,091 &   52.49       & 443   &448\\
         \hline
         eu-2015    &   6,650,532   &   165,693,531 &     24.91      & 9,568 & 9,569\\
         \hline
         arabic     &   22,744,080  &   639,999,458 & 28.14 &3,126  &3,126  \\
         \hline
         it-2004    &   41,291,594  &  1,150,725,436  &   27.86  & 3,198 &3,197\\
         \hline
         sk-2005    &   50,636,154  &   1,949,412,601 &  38.50 &4,502  &4,502\\
         \hline
         uk-2007    &   110,123,614 & 3,944,932,566 &   35.82  &10,027 &10,027  \\
         \hline
    \end{tabular}
    }
\end{table}

\begin{figure*}[ht]
\hspace{-0.35in}
\centering
\begin{tabular}{c c c c c c}{}
  \begin{minipage}{2.6cm}
  \includegraphics[width=2.8cm]{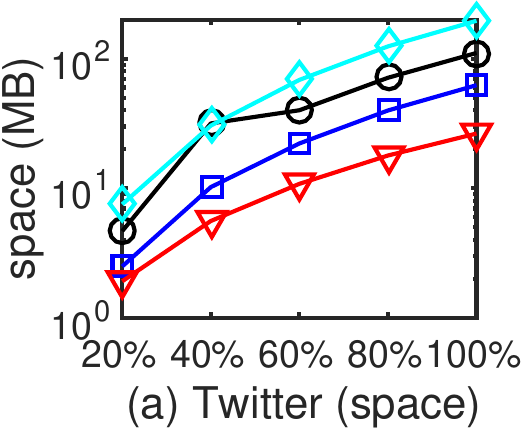}
  \end{minipage}
  &
  \begin{minipage}{2.6cm}
  \includegraphics[width=2.8cm]{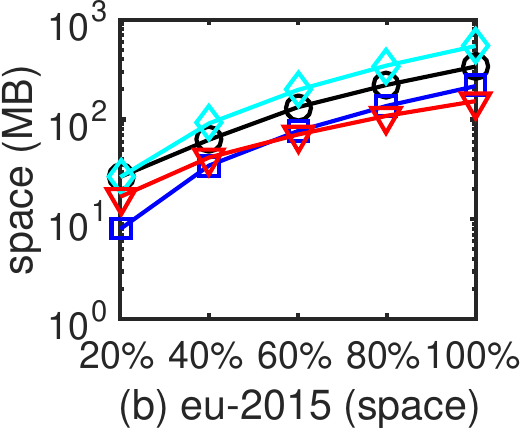}
  \end{minipage}
  &
  \begin{minipage}{2.6cm}
  \includegraphics[width=2.8cm]{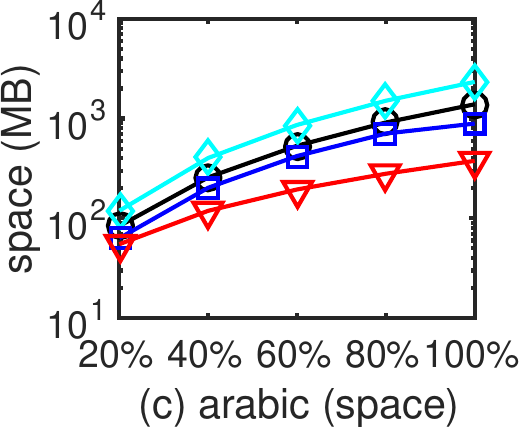}
  \end{minipage}
  &
  \begin{minipage}{2.6cm}
  \includegraphics[width=2.8cm]{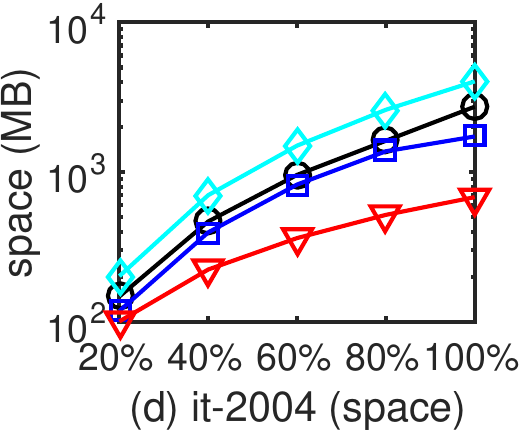}
  \end{minipage}
  &
  \begin{minipage}{2.6cm}
  \includegraphics[width=2.8cm]{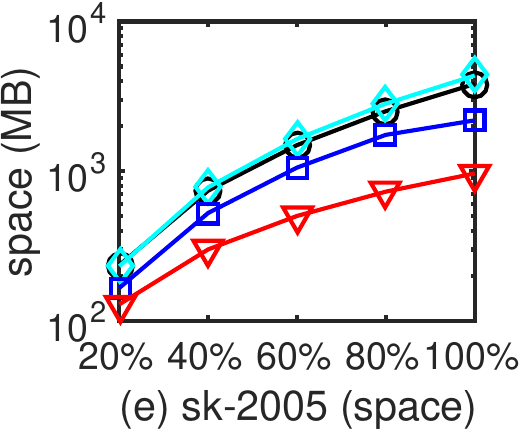}
  \end{minipage}
  &
   \begin{minipage}{2.6cm}
  \includegraphics[width=2.8cm]{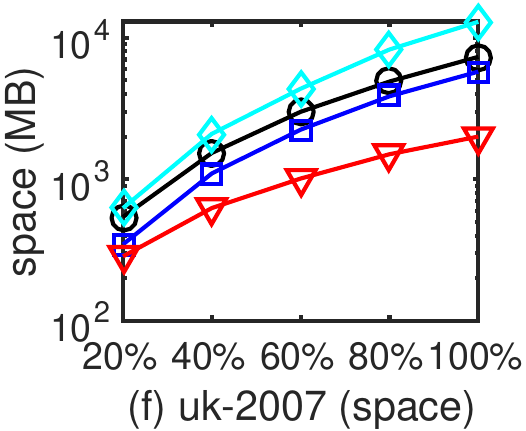}
  \end{minipage}
  \\
  &
  &
  &
  \hspace{-2in}
  \begin{minipage}{2.6cm}
  \includegraphics[width=5cm]{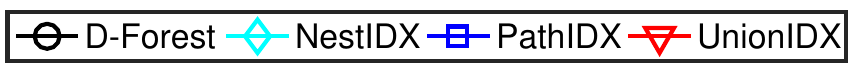}
  \end{minipage}
  \\
  \begin{minipage}{2.6cm}
  \includegraphics[width=2.8cm]{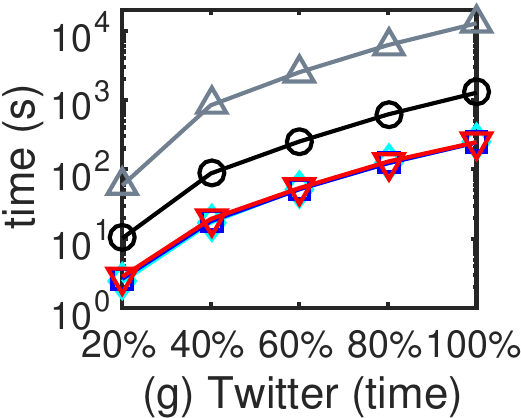}
  \end{minipage}
  &
  \begin{minipage}{2.6cm}
  \includegraphics[width=2.8cm]{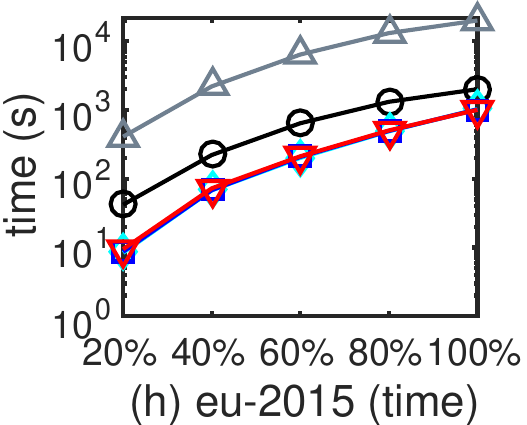}
  \end{minipage}
  &
  \begin{minipage}{2.6cm}
  \includegraphics[width=2.8cm]{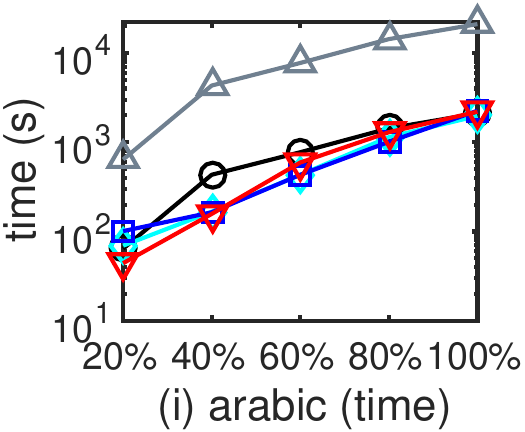}
  \end{minipage}
  &
  \begin{minipage}{2.6cm}
  \includegraphics[width=2.8cm]{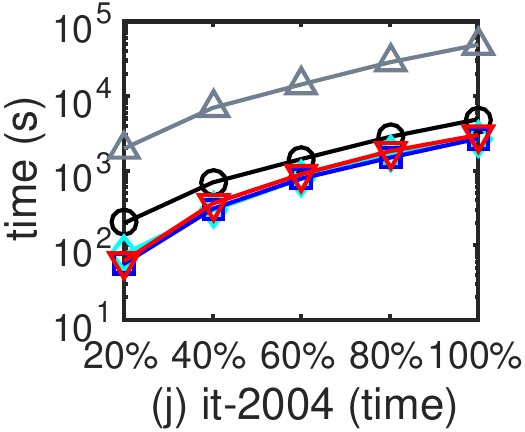}
  \end{minipage}
  &
  \begin{minipage}{2.6cm}
  \includegraphics[width=2.8cm]{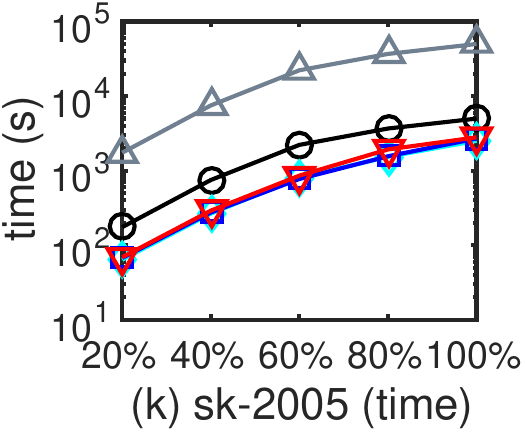}
  \end{minipage}
  &
   \begin{minipage}{2.6cm}
  \includegraphics[width=2.8cm]{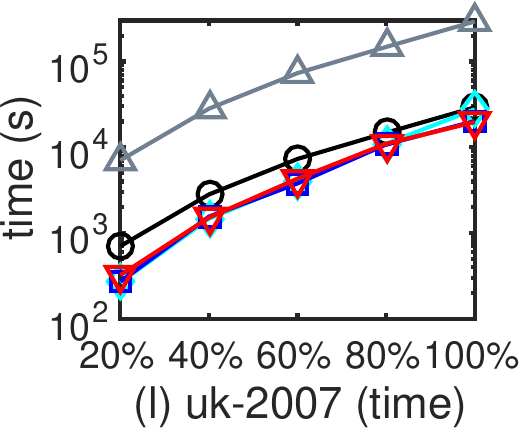}
  \end{minipage}
  \\
  &
  &
  &
  \hspace{-2.7in}
  \begin{minipage}{2.6cm}
  \includegraphics[width=6.8cm]{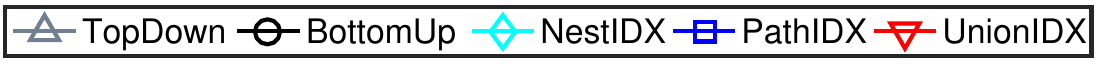}
  \end{minipage}

\end{tabular}
\caption{The space cost of D-Forest and the time cost of index construction.}
\label{fig:exp-index}
\end{figure*}

\begin{figure*}[ht]
\hspace{-0.4in}
\centering
\begin{tabular}{c c c c c c}
  \begin{minipage}{2.6cm}
  \includegraphics[width=2.9cm]{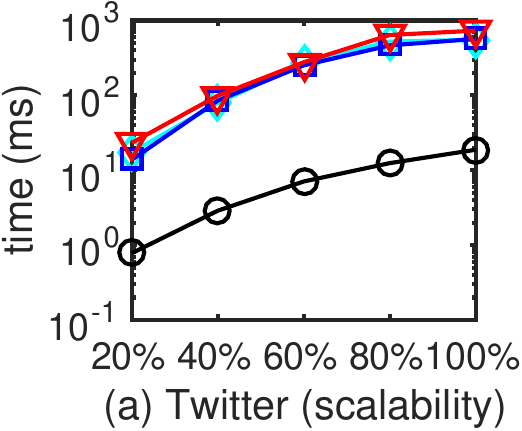}
  \end{minipage}
  &
  \begin{minipage}{2.6cm}
  \includegraphics[width=2.8cm]{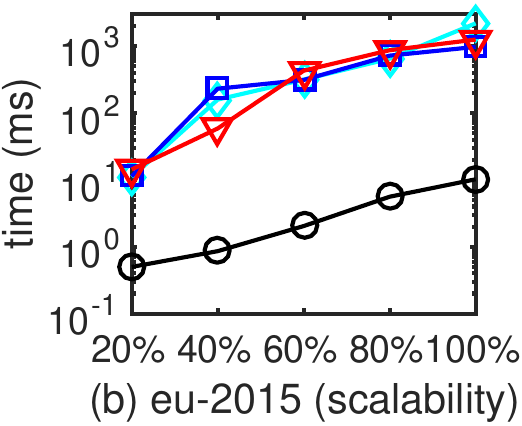}
  \end{minipage}
  &
  \begin{minipage}{2.6cm}
  \includegraphics[width=2.9cm]{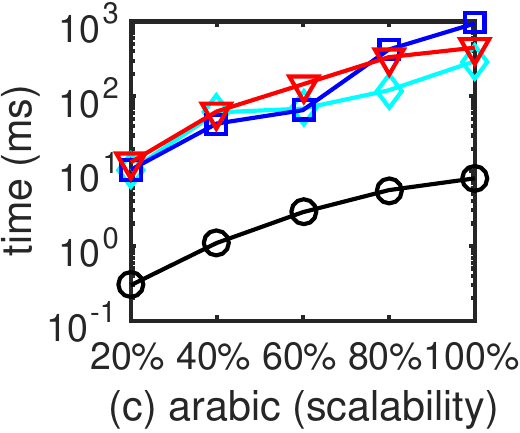}
  \end{minipage}
  &
   \begin{minipage}{2.6cm}
  \includegraphics[width=2.9cm]{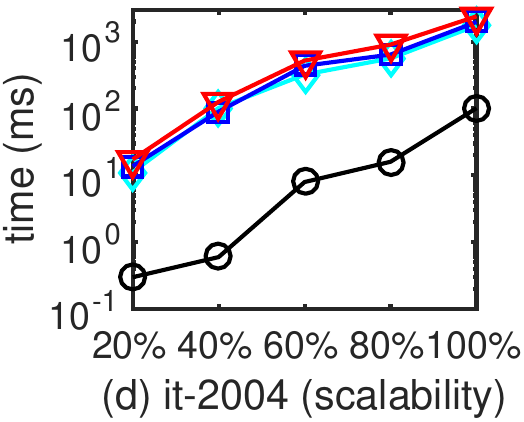}
  \end{minipage}
  &
   \begin{minipage}{2.6cm}
  \includegraphics[width=2.9cm]{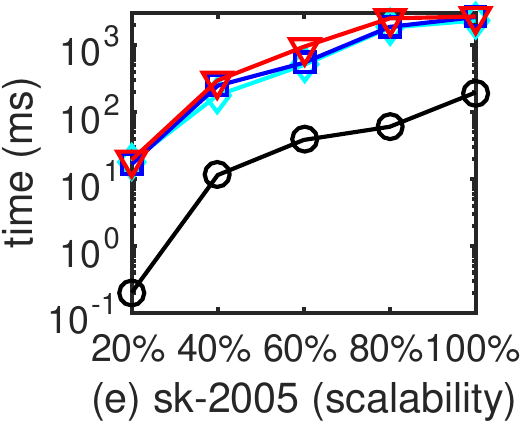}
  \end{minipage}
  &
   \begin{minipage}{2.6cm}
  \includegraphics[width=2.9cm]{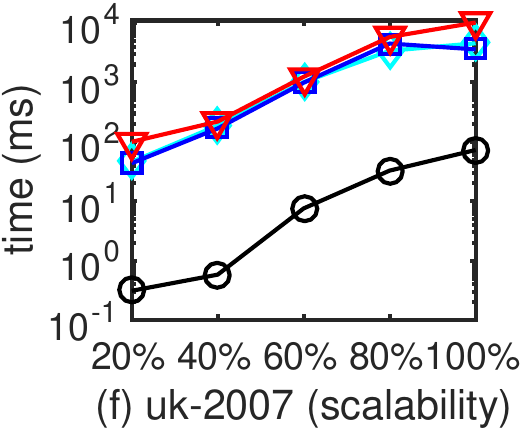}
  \end{minipage}
  \\
  \begin{minipage}{2.6cm}
  \includegraphics[width=2.8cm]{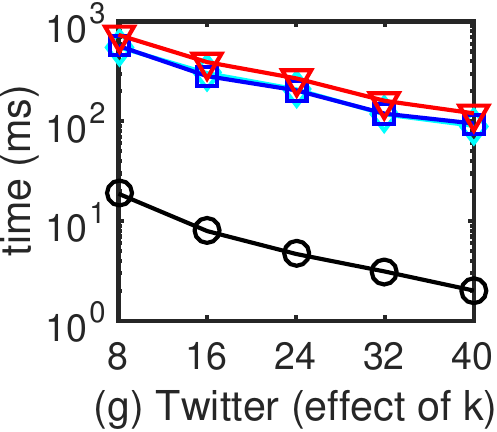}
  \end{minipage}
  &
  \begin{minipage}{2.6cm}
  \includegraphics[width=2.8cm]{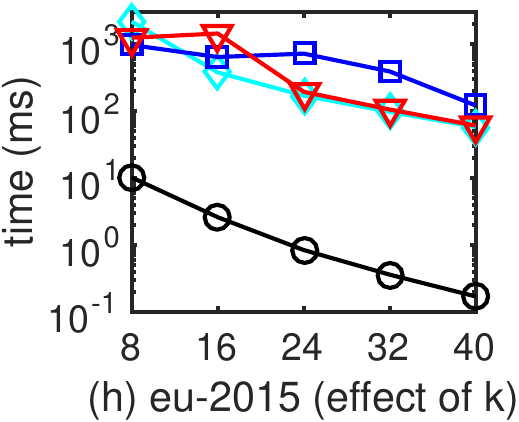}
  \end{minipage}
  &
  \begin{minipage}{2.6cm}
  \includegraphics[width=2.8cm]{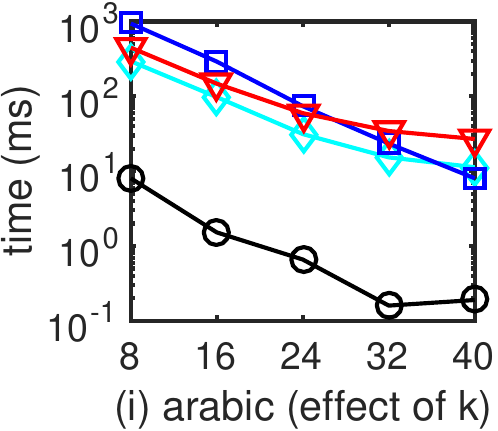}
  \end{minipage}
  &
   \begin{minipage}{2.6cm}
  \includegraphics[width=2.8cm]{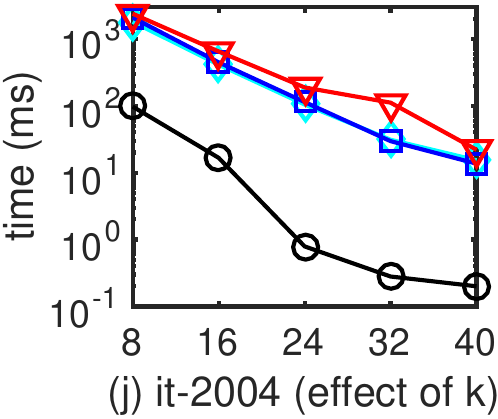}
  \end{minipage}
  &
   \begin{minipage}{2.6cm}
  \includegraphics[width=2.8cm]{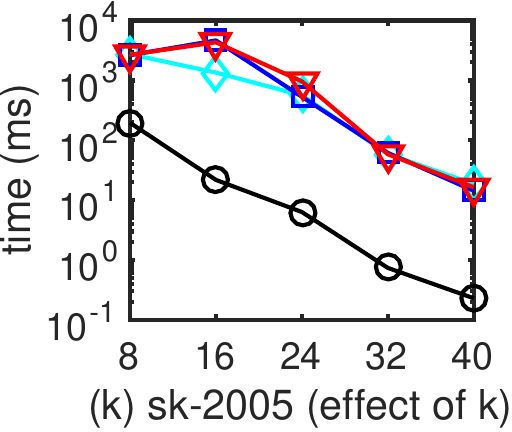}
  \end{minipage}
  &
   \begin{minipage}{2.6cm}
  \includegraphics[width=2.8cm]{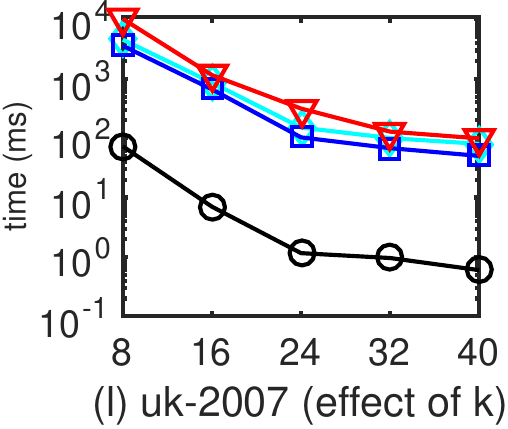}
  \end{minipage}
  \\
  \begin{minipage}{2.6cm}
  \includegraphics[width=2.8cm]{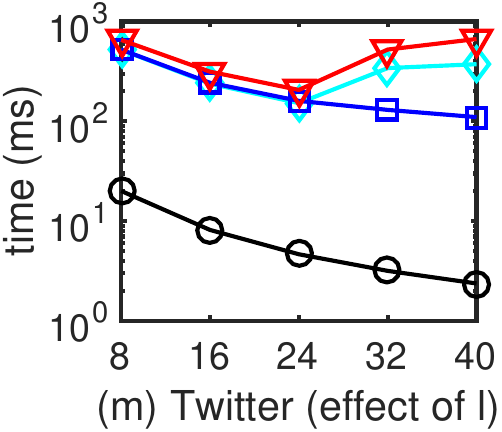}
  \end{minipage}
  &
  \begin{minipage}{2.6cm}
  \includegraphics[width=2.8cm]{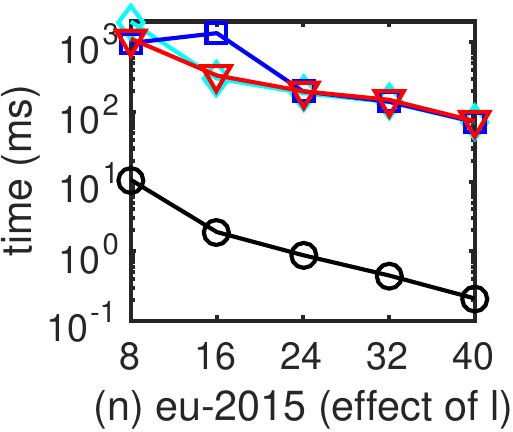}
  \end{minipage}
  &
  \begin{minipage}{2.6cm}
  \includegraphics[width=2.8cm]{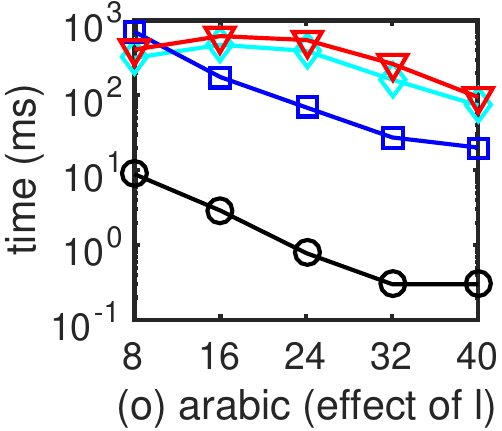}
  \end{minipage}
  &
   \begin{minipage}{2.6cm}
  \includegraphics[width=2.8cm]{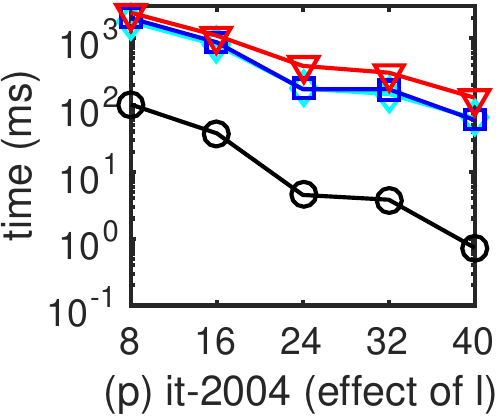}
  \end{minipage}
  &
   \begin{minipage}{2.6cm}
  \includegraphics[width=2.8cm]{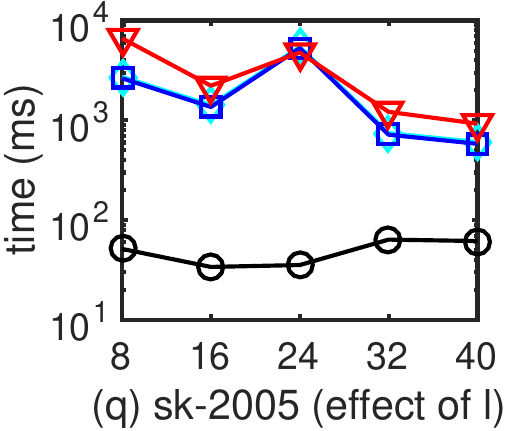}
  \end{minipage}
  &
   \begin{minipage}{2.6cm}
  \includegraphics[width=2.8cm]{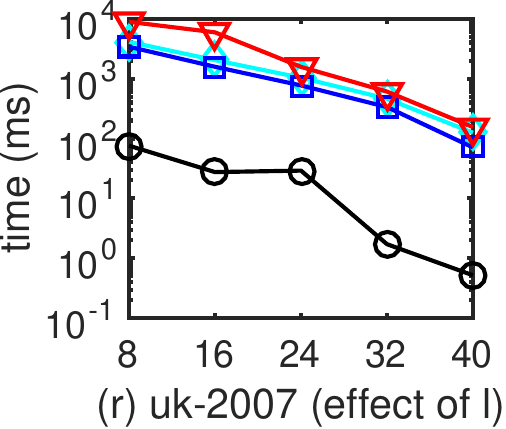}
  \end{minipage}
  \\
  &
  &
  &
  \hspace{-1.85in}
\begin{minipage}{2.6cm}
  \includegraphics[width=5cm]{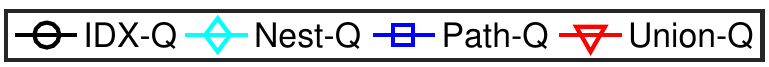}
  \end{minipage}
\end{tabular}
\caption{The efficiency of answering CSD queries.}
\label{fig:exp-query}
\end{figure*}

\subsection{Experimental Results}
\label{sec:results}
We compare our approach with the state-of-the-art solutions~\cite{fang2019effective} which propose three indices, i.e., {\tt NestIDX}, {\tt PathIDX} and {\tt UnionIDX}, and three corresponding query algorithms, i.e., {\tt Nest-Q}, {\tt Path-Q} and {\tt Union-Q}. Using the same evaluation strategy in~\cite{fang2019effective}, we respectively report the efficiency results of index construction and queries in Figures~\ref{fig:exp-index} and \ref{fig:exp-query}. To evaluate the efficiency of index construction, for each dataset, we randomly select 20\%, 40\%, 60\% and 80\% of its vertices and obtain the four subgraphs induced by these vertices. 

\noindent{\bf\underline{1. Space cost of D-Forest.}}
To measure the space cost of an index, we store all the index elements, which can be used to recover the index, into the disk. As shown in Figure~\ref{fig:exp-index} (a)-(f), with the sizes of sub-datasets growing, the space cost of D-Forest and others increase steadily. Besides, D-Forest uses comparable space cost as other indexes, meaning that D-Forest is as space-efficient as the state-of-the-art solutions.

\noindent{\bf\underline{2. Time cost of index construction.}}
From Figure~\ref{fig:exp-index}(g)-(l), we can observe that {\tt BottomUp} takes similar time cost with state-of-the-art solutions . For instance, on the largest dataset uk-2007 (35.4GB in disk), {\tt BottomUp} takes only 8.12 hours to build D-Forest. And {\tt BottomUp} always runs at least 10 times faster than {\tt TopDown}. We remark that to save the computational resources, for each test, we terminate {\tt TopDown} if it runs 10 times longer than {\tt BottomUp}.

\noindent{\bf\underline{3. Scalability evaluation of CSD queries.}}
In Figure~\ref{fig:exp-query}(a)-(f), we evaluate the scalability of CSD query algorithms over different sizes of datasets. As suggested in~\cite{fang2019effective}, for each sub-dataset, we randomly select 200 vertices which are in the (8,8)-cores as the query vertices and set $k$=$l$=$8$. Generally, the running time of all algorithms increases as the size of sub-datasets grows. Besides, our query algorithm {\tt IDX-Q} achieves the best performance on all the datasets. For example, on uk-2007 dataset with around 4 billion edges, {\tt IDX-Q} takes 72ms on average to answer a query and runs about 100 times faster than the existing solutions.

\noindent{\bf\underline{4. Effect of $k$ and $l$ in CSD queries.}}
We depict the effect of $k$ and $l$ on the efficiency of CSD queries in Figure~\ref{fig:exp-query}(g)-(r). For each dataset, we randomly select 200 vertices within the (8,8)-cores to query. We see that as the value of $k$ and $l$ increases, the returned communities become smaller, so the time cost of all query algorithms decreases. Again, {\tt IDX-Q} is up to two orders of magnitude faster than the three existing algorithms which generally achieve similar performance.

\noindent{\bf\underline{5. Efficiency of SCSD queries.}}
Generally, SCSD queries are slower than CSD queries, since they have to spend extra time cost to compute SCC's.
We have compared {\tt IDX-SQ} with {\tt Nest-SQ}, {\tt Path-SQ} and {\tt Union-SQ} proposed in~\cite{fang2019effective}, and the results show that {\tt IDX-SQ} can run about 4 times faster than these three algorithms. For example, on uk-2007, {\tt IDX-SQ} takes 1.85s on average to find an SCC-based (8, 32)-core, while others take 3.21s, 2.38s and 6.88s. For lack of space, we omit the detailed results here.

%% file: conclusion.tex
\section{Conclusion}
\label{sec:conclusion}
In this paper, we examine the CSD problem and design a novel index D-Forest, based on which a CSD query can be completed in the optimal time cost. To build the index, we propose a basic algorithm {\tt TopDown}, and an advanced algorithm {\tt BottomUp} by introducing the CUF data structure. 
We also discuss the index maintenance for dynamic graphs and how to answer the SCSD query using D-Forest. 
The experimental results on six real large graphs demonstrate the efficiency of our solutions.

In the future, we will investigate how to extend other cohesiveness metrics, such as $k$-truss and $k$-clique, to search communities in large directed graphs.
We will also study how to find meaningful communities from directed graphs with attributes, e.g., keywords and locations.

%% file: appdendix.tex
\section{Proofs of Lemmas}
\label{appendixA}

\addtocounter{lemma}{-6}

\begin{lemma}[Space cost]
\label{lm:space}
Given a directed graph $G$, its D-Forest takes $O(m)$ space.
\end{lemma}
\begin{proof}
For each vertex $v$, if its in-degree is $k$, then it appears in at most $k$ trees, and in each tree, it appears only twice (one in the tree and one in the auxiliary map). Thus, the space cost of $v$ is bounded by $O(deg_G^{in}(v))$. Hence, Lemma \ref{lm:space} holds.
\end{proof}

\begin{lemma}[Query cost]
\label{lm:query}
Given a D-Forest, {\tt IDX-Q} completes in the optimal time and space cost, i.e., $O(|C|)$.
\end{lemma}
\begin{proof}
The lemma directly follows the observation.
\end{proof}

\begin{lemma}
\label{lemma:topdown}
Given a directed graph $G$, the time cost of building D-Forest using {\tt TopDown} is $O(m^2)$.
\end{lemma}
\begin{proof}
Given a specific $k$, computing all the ($k$, $l$)-cores from the ($k$, $0$)-core takes $O(m)$ time; besides, searching all the ($k$, $l$)-$\widehat{core}$s from the ($k$, $l$$-$$1$)-$\widehat{core}$s takes $O(m)$~\cite{fang2019effective}.
Thus, it takes $O(l_{\max}\cdot m)$ time to build the $k$-tree.
Since there are ($k_{\max}$$+$$1$) trees, the total cost is $O(k_{\max}\cdot l_{\max}\cdot m)$.
Meanwhile, $k_{\max}$ and $l_{\max}$ are at most $(\sqrt{4m+1} - 1)/2$ \cite{fang2019effective}, so Lemma \ref{lemma:topdown} holds.
\end{proof}

\begin{lemma}[Space cost]
Given a directed graph $G$, the CUF data structure of all vertices costs $O(n)$ space.
\end{lemma}
\begin{proof}
The lemma directly follows the observation.
\end{proof}

\begin{lemma}[Time cost of CUF functions]
\label{lm:cuf}
\textsc{makeSet} takes $O(1)$ time; for \textsc{union} and \textsc{find}, the amortized time per operation is $O(\alpha(n))$; \textsc{updateCuf} takes $O(\alpha(n) \cdot |V|)$ time.
\end{lemma}
\begin{proof}
Obviously, \textsc{makeSet} for each vertex takes $O(1)$. As for \textsc{union} and \textsc{find}, since they use \emph{union by rank} and \emph{path compression} optimization, the amortized time per operation is $O(\alpha(n))$~\cite{tarjan1979class}. Thus, \textsc{updateCuf} totally takes $O(\alpha(n) \cdot |V|)$ for all vertices in $V$.
\end{proof}

\begin{lemma}
\label{lm:bottomup}
Given a directed graph $G$, the time cost of building D-Forest using {\tt BottomUp} is $O(\alpha(n)\cdot m \cdot \sqrt{m})$.
\end{lemma}
\begin{proof}
Decomposing the ($k$, $0$)-core for a specific $k$ takes $O(m)$ time~\cite{fang2019effective}.
Then \textsc{buildALevel} takes $O(\alpha(n) \cdot m_l)$, where $m_l$ is the edges visited in the $l$-th level. This implies that building $k$-tree takes $O(\alpha(n)\cdot m)$ in total. 
Thus {\tt BottomUp} takes $O(k_{\max}\cdot \alpha(n)\cdot m)$, which is upper bounded by $O(\alpha(n)\cdot m \cdot \sqrt{m})$. Lemma~\ref{lm:bottomup} holds.
\end{proof}


